\documentclass{article}
\usepackage{amssymb}
\usepackage{amsthm} 
 \usepackage{proof}
 \usepackage{fullpage}

\usepackage[cmex10]{amsmath}
\usepackage{array}
\usepackage{url}
\usepackage{cite}
\usepackage{color}
\usepackage{proof}
\usepackage{graphicx}
\usepackage{xspace}
\usepackage{subfigure}
\usepackage{mycommands}

\usepackage[hide]{ed}

\begin{document}


\title{Hybrid and Subexponential Linear Logics\\ Technical Report}

\author{Jo{\"e}lle Despeyroux\footnote{INRIA and CNRS, I3S, Sophia-Antipolis, France. (joelle.despeyroux@inria.fr)}, Carlos Olarte and Elaine Pimentel\footnote{Universidade Federal do Rio Grande do Norte. Brazil.  (carlos.olarte@gmail.com, elaine.pimentel@gmail.com)}}
        
\date{{\em Research Report HAL nb} hal-01358057 --- September 1, 2016 } 

        \maketitle

\begin{abstract}
HyLL (Hybrid Linear Logic) and SELL (Subexponential Linear Logic)  are logical frameworks 
that have been extensively used for specifying systems that  
exhibit modalities such as temporal or spatial ones.
Both frameworks
have 
linear logic (LL) as 
a common ground 
 and they admit  (cut-free)
complete focused proof systems. The difference between the two logics  relies on the way modalities are handled. In HyLL, truth judgments are labelled by {\em worlds} and  
hybrid connectives relate worlds with
formulas.
 In SELL, the linear logic exponentials ($\bang$, $\quest$) are decorated with labels representing {\em locations}, and an ordering on such labels 
 defines the provability relation among resources in those locations.
It is well known that SELL, as a logical framework, 
is strictly more expressive than LL. However, so far, it was not clear whether HyLL is more expressive than LL and/or SELL. 
In this paper, we show an encoding of the  HyLL's logical rules  into LL with the highest level of adequacy, hence showing that
HyLL is as expressive as LL.
 We also propose an encoding of HyLL into \sellU\ (\sell\ plus quantification over locations) that gives better insights about the meaning of worlds in HyLL. 
%
We conclude our expressiveness study  by showing that previous attempts of encoding 
Computational Tree Logic (CTL)  operators into HyLL cannot be extended to consider the whole set of temporal connectives. We show that a system of  LL with fixed points is indeed needed to faithfully encode the behavior of such temporal operators. 

\end{abstract}

\section{Introduction} \label{sec:intro}

Logical
frameworks are 
adequate tools for specifying proof systems, since they support levels of abstraction that facilitate writing declarative specifications of object-logic proof systems. 
Many frameworks have  been 
used  for the specification of  proof systems, and linear logic~\cite{girard87tcs} (LL) is one of the most successful ones. This is mainly because   LL is resource conscious and, at the same time, it can internalize classical and intuitionistic behaviors (see, for example,~\cite{DBLP:journals/tcs/MillerP13,cervesato02ic}).

However, since specifications of object-level systems into the logical framework should be natural and direct, there are some features that cannot be adequately captured in LL, in particular modalities different from the ones present in LL. 

Extensions of LL have been proposed in order to fill this gap. 
The aim is to propose  stronger logical frameworks that preserve the elegant properties of  linear logic as the underlying logic.
Two of such extensions are
HyLL (Hybrid Linear Logic)\footnote{Actually, HyLL is an extension of intuitionistic linear logic (ILL), while \sell\ can be viewed as an extension of both ILL or LL.} \cite{ChaudhuriDespeyroux:14} and SELL (Subexponential Linear Logic) \cite{danos93kgc,OlartePimentelNigam:tcs-15}. These logics  have
been  extensively
used for specifying systems that  
exhibit modalities such as temporal or spatial ones.
The difference between HyLL and SELL relies on the way modalities are handled. 

In HyLL, truth judgments are labeled by worlds and  two hybrid connectives relate worlds with
formulas: the satisfaction $\at$ which states that a proposition is true at a given world, and the localization $\da$ which binds a name for the (current) world the proposition
is true at. These constructors
allow for the specification of modal connectives such as $\Box A$ ($A$ is true in all the accessible worlds)  and $\Diamond A$ (there exists an accessible world where $A$ holds). The underlying structure on worlds allows for the  modeling of  transitions systems and the specification of temporal formulas \cite{ChaudhuriDespeyroux:14,deMaria-Despeyroux-Felty:14-fmmb}.

In SELL, the LL exponentials ($\bang$, $\quest$) are decorated with labels:  the formula 
$\nquest{a}A$ 
can be interpreted as $A$ {\em holds in a
 location}, \emph{modality}, or \emph{world} $a$. Moreover, $A$ can be deduced in a location $b$ related to $a$ ($b \preceq a$). On the other side,  the formula  $\nquest{a}\nbang{a}{A}$ means that  $A$ {\em is confined into the location} $a$, that is, the information $A$ is not propagated to other worlds/locations related to $a$. 
While linear logic has only seven logically distinct prefixes of bangs and question-marks (none, $!$, $?$, $!?$, $?!$, $!?!$, $?!?$), SELL allows for an unbounded number of such prefixes (e.g., $\nbang{a}\nquest{c}\nquest{d}$). Hence SELL 
enhances the expressive power of LL as a logical framework. 

Up to now, it was not clear how HyLL is related to LL and/or SELL. In this paper we answer that question by 
showing a direct encoding of the  HyLL's logical rules  into LL with the highest level of adequacy. Hence, we show that 
HyLL is actually as expressive as LL.

 We also propose an encoding of HyLL into \sellU\ (\sell\ with quantification over locations) that gives better insights about the meaning of worlds in HyLL.  
More precisely, we represent HyLL formulas as formulas in \sell\ and 
encode the logical rules as formulas in \sellU. We show that a flat  subexponential structure is sufficient  for representing  any world structure in HyLL. 
This explains better why the worlds in HyLL do not add any expressive power to LL: 
they cannot control the logical context as the subexponentials do with the promotion rule.

HyLL has been shown to be a flexible framework for the specification of biological systems
\cite{deMaria-Despeyroux-Felty:14-fmmb}
where both the system and its properties 
are specified using the same  logic. More precisely, the properties  of interest are first written in 
Computational Tree Logic (CTL) and later encoded as HyLL formulas. 
However, there was no a formal statement about the CTL fragments that can be adequately captured in HyLL. 
Hence,  the last contribution of this paper is to continue our study of HyLL theory and  to push forward previous attempts of using this logic  for the specification of transition systems and formulas in CTL.  We show that it is not possible to adequately encode, in HyLL,  the universal path  quantifier $\tA$ (for all paths), nor the temporal formula $\tE\tG Q$ (there exists a path where $Q$ always holds). 
The definition of such formulas is recursive, hence 
one needs to use induction, at the meta-level,  to accurately capture their behavior.  Instead of using meta-reasoning, as done in \cite{deMaria-Despeyroux-Felty:14-fmmb}, we use a logical framework featuring fixed point constructs. More precisely, we use  additive multiplicative LL 
with fixed point operators 
($\mu$MALL) \cite{DBLP:journals/tocl/Baelde12}
 for  the encoding of CTL. 
 We show that the well known fixed point characterization of CTL~\cite{DBLP:journals/iandc/BurchCMDH92}  can be  matched by the fixed point operators of 
$\mu$MALL.

The rest of the paper is organized as follows. We briefly recall LL in Section \ref{sec:ll}, HyLL in Section \ref{sec:hyll} and SELL in Section \ref{sec:sell}. 
The encoding of HyLL logical rules into LL is discussed in 
Section \ref{sec:hyll-ll}.  Section \ref{sec:hyllsell} 
presents 
the encoding of HyLL into \sellU. We also prove that information confinement, a feature in SELL that is needed to specify spatial systems, cannot be captured in HyLL.  Section 
\ref{sec:temporal}  shows how to encode CTL   into $\mu$MALL. Section 
\ref{sec:conclusion} concludes the paper.



\section{Preliminaries} 
Although we assume that the reader is familiar with linear logic~\cite{girard87tcs} (LL), we
review some of its basic proof theory in the following sections.

\subsection{Linear Logic and Focusing} \label{sec:ll}
\emph{Literals} are either atomic formulas ($p$) or their
negations ($p^\bot$).  The connectives $\tensor$ and $\lpar$ and their units $1$
and $\bot$ are \emph{multiplicative}; the connectives $\plus$ and
$\with$ and their units $0$ and $\top$ are \emph{additive}; 
$\forall$ and $\exists$ are (first-order) quantifiers;
and $\bang$ and $\quest$ are the exponentials (called bang and question-mark,
respectively).  

First proposed by Andreoli \cite{DBLP:journals/logcom/Andreoli92} for linear logic,
focused proof systems provide  normal form proofs for cut-free proofs.
The connectives of linear logic can be divided into two classes.  The
{\em negative} connectives have invertible introduction rules: these
connectives are $\lpar$, $\bottom$, $\with$, $\top$, $\forall$, and
$\quest$.  The {\em positive} connectives 
$\otimes$, $\one$, $\oplus$,
$\zero$, $\exists$, and $\bang$
are the de Morgan duals of
the negative connectives.
  A formula is {\em positive} if it is
a negated atom or its top-level logical connective is positive.
Similarly, a formula is {\em negative} if it is an atom or its
top-level logical connective is negative. 

Focused proofs are organized into two \emph{phases}.  In the \emph{negative} phase, all the invertible inference rules are eagerly applied. 
The \emph{positive} phase begins by choosing  a positive formula $F$ on which to focus. Positive rules are applied to $F$ until either $\one$ or a negated atom is encountered (and the proof must end by applying the initial rules) or
the promotion rule ($\bang$) is applied  or a negative subformula is encountered, when the proof switches to the negative phase.

This change of phases on proof search is particularly interesting when the focused
formula is a {\em bipole}~\cite{DBLP:journals/logcom/Andreoli92}.
\begin{definition}
We call a {\em monopole} a linear logic formula that is built up from atoms and occurrences of the negative connectives, with the restriction that $\quest$  has atomic scope. 
{\em Bipoles}, on the other hand, are
positive formulas built from monopoles and negated atoms using only positive connectives, with the additional restriction that $\bang$  can only be applied to a monopole.
\end{definition}
Focusing on a bipole 
will produce a
single positive and a single negative phase. This two-phase decomposition 
enables us to  adequately capture  the application of object-level inference rules by the meta-level linear logic, as will be shown in Section~\ref{sec:hyll-sell}.

The focused system LLF for classical linear logic 
can be found in the appendix.

\subsection{Hybrid Linear Logic} \label{sec:hyll}
Hybrid Linear Logic  (HyLL)
is a conservative extension of intuitionistic
first-order linear logic (ILL)~\cite{girard87tcs} where the truth judgments 
are labeled by worlds representing constraints on states and state transitions. 
Judgments of HyLL are of the form 
``$A$ is true at world $w$'', abbreviated as $A ~@~ w$. 
Particular choices of worlds
produce particular instances of HyLL. 
Typical examples  are ``$A$ is true at time $t$'', 
or ``$A$ is true with probability $p$''.
HyLL was first proposed in~\cite{ChaudhuriDespeyroux:14} and it has  been used 
as a logical framework for specifying biological
systems~\cite{deMaria-Despeyroux-Felty:14-fmmb}.

Formally, worlds are defined as follows.
\begin{definition}[HyLL worlds] \label{def:constraint-domain}
  A \emph{constraint domain} $\cal W$ is a monoid structure $\langle W, ., \iota\rangle$. 
  The elements of $W$ are called \emph{worlds}
  and its \emph{reachability relation}  $\preceq\ : W \times W$ is defined as  $u \preceq w$ if there exists 
  $v \in W$ such that $u . v = w$. 
\end{definition}
\noindent
The identity world $\iota$ is $\preceq$-initial and it is intended to represent the
lack of any constraints. Thus, the ordinary first-order linear logic is embeddable into any 
instance of HyLL by setting all world labels to the identity.
A typical, simple example of constraint domain is 
$\mathcal{T} = \langle \N, +, 0\rangle$, representing instants of time. 

Atomic propositions $(p,q,\ldots)$ are applied to a sequence of terms 
$(s,t,\ldots)$, which are drawn from an untyped term language containing 
constants  $(c,d,\ldots)$, term variables $(x, y, \ldots)$ and function symbols 
$(f, g, \ldots)$ applied to a list of terms $(\vec{t})$. 
Non-atomic propositions are constructed from the connectives of first-order
intuitionistic linear logic and the two hybrid connectives. Namely,  
\emph{satisfaction} ($\texttt{at}$), which states that a
proposition is true at a given world ($w, \iota, u.v, \ldots$), and
\emph{localization} ($\downarrow$), which binds a name for the (current) world 
the proposition is true at. 
The following grammar summarizes the syntax of HyLL.

\begin{tabular}{l@{\ }r@{\ }l}
  $t$ & $::=$ & 
       $c ~|~ x ~|~ f(\vec t)$ \\
  $A, B$ & $::=$ & 
       $p(\vec t) ~|~ A \otimes B ~|~ \mathbf{1} ~|~ A \rightarrow B ~|~ 
                A \mathbin{\&} B ~|~ \top ~|~
       A \oplus B ~|~  \mathbf{0} ~|~ ! A ~|~ $  \\
       & &
       $\forall x.~ A ~|~ \exists x.~ A ~|~ $

       $(A ~\at~ w) ~|~ \downarrow u.~ A ~|~ \forall u.~ A ~|~ \exists u.~ A$ \\
\end{tabular}

\noindent
Note that world $u$ is bounded in 
the propositions $\downarrow u.~A$, $\forall u.~ A$ and $\exists u.~ A$.
World variables cannot be used in terms, and neither can term variables occur in 
worlds. This restriction is important for the modular design of HyLL because it 
keeps purely logical truth separate from constraint truth.  
%
We note that $\downarrow$ and $\at$ commute freely with all non-hybrid 
connectives~\cite{ChaudhuriDespeyroux:14}.  


The sequent calculus  \cite{Gentzen35} presentation of HyLL uses sequents of
the form $\Gamma; \Delta \vdashseq C ~@~ w$ where 
$\Gamma$ (\emph{unbounded context}) is a set and 
$\Delta$ (\emph{linear context}) is a multiset
of judgments of the form $A ~@~ w$.
Note that in a judgment $A ~@~ w$ (as in a proposition $A ~\at~ w$), $w$ can be 
any expression in $\cal W$, not only a variable.

The inference rules dealing with the new hybrid connectives are
depicted below 
(the complete set of rules can be found in the appendix).
$$
  \dfrac{\Gamma ; \Delta \vdashseq A @ u} 
        {\Gamma ; \Delta \vdashseq (A ~\at~ u) @ w} \at R
  ~ \qquad
  \dfrac{\Gamma ; \Delta, A @ u \vdashseq C @ w} 
        {\Gamma ; \Delta, (A ~\at~ u) @ v \vdashseq C @ w} \at L
$$
$$
  \dfrac{\Gamma ; \Delta \vdashseq A [w / u] @ w} 
        {\Gamma ; \Delta \vdashseq \downarrow u. A @ w} \downarrow R
  ~  \qquad
  \dfrac{\Gamma ; \Delta, A [v / u] @ v \vdashseq C @ w}
        {\Gamma ; \Delta, \downarrow u. A @ v \vdashseq C @ w} \downarrow L
$$
Note that $(A ~\at~ u)$ is a \emph{mobile} proposition:
it carries with it the world at which it is true. 
%
%
Weakening and contraction are admissible rules for the unbounded context.  

The most important structural properties are the admissibility of 
the general identity (i.e. over any formulas, not only atomic propositions) and cut theorems. 
While the first provides a syntactic completeness theorem for the logic, 
the latter guarantees consistency
(i.e. that there is no proof of  $ . ; . \vdashseq \mathop{0} ~@~ w$).
%

\begin{theorem}[Identity/Cut] 
  \label{thm:cut} 
  $
   \\
   1.~\Gamma ; A ~@~ w \vdashseq A ~@~ w\\
   2.~ {If}~ \Gamma ; \Delta \vdashseq A ~@~ u 
   ~{and}~ \Gamma ; \Delta', A ~@~ u \vdashseq C ~@~ w, 
   ~{then}~ \Gamma ; \Delta, \Delta' \vdashseq C ~@~ w \\
   3.~ {If}~ \Gamma ; . \vdashseq A ~@~ u 
   ~{and}~ \Gamma, A ~@~ u ; \Delta \vdashseq C ~@~ w, 
   ~{then}~ \Gamma ; \Delta \vdashseq C ~@~ w.
$
\end{theorem}

HyLL is conservative with respect to intuitionistic linear logic: as long as
no hybrid connectives are used, the proofs in HyLL are identical to those in ILL.
%
Moreover,  HyLL 
is more expressive than S5, as it allows direct manipulation of the worlds using
the hybrid connectives, while HyLL's $\delta$ connective 
(see Section~\ref{sec:temporal}) is not definable in S5.
We also note that HyLL admits a complete focused  \cite{DBLP:journals/logcom/Andreoli92} proof system.
The interested reader can find proofs and further meta-theoretical theorems 
about HyLL in~\cite{ChaudhuriDespeyroux:14}.

%
%

\subsection{Subexponentials in Linear Logic}
\label{sec:sell} 
Linear logic with subexponentials\footnote{We  note  that intuitionistic and classical \sell\ are equally expressive, as shown in~\cite{DBLP:conf/csl/Chaudhuri10}. Hence, although we will introduce here the classical version of \sell\ (needed in Section~\ref{sec:hyllsell}), we could also present \sell\ as an extension of ILL.} (\sell) shares with LL  all 
its connectives except the exponentials:
instead of having a single pair of exponentials $\bang$ and $\quest$, \sell\ may
contain as many \emph{subexponentials}~\cite{danos93kgc,nigam09ppdp,OlartePimentelNigam:tcs-15}, written $\nbang{a}$  and $\nquest{a}$,
as one needs.  
The grammar of formulas in  \sell\  is as follows: 
\[\begin{array}{lcl}
 F &::=& \zero \mid \one \mid  \top \mid \bottom \mid p(\vec t) \mid F_1 \tensor F_2 \mid F_1 \oplus F_2 \mid F_1 \lpar F_2 \mid  F_1 \with F_2 \mid \\
& & \exists x. F \mid \forall x. F\mid  
  \nbang{a} F \mid \nquest{a} F 
 \end{array}
\]
The proof system for \sell\ is specified by a
\emph{subexponential signature} $\Sigma = \tup{I, \preceq, U}$, where $I$
is a set of labels, $U \subseteq I$ is a set specifying which
subexponentials allow weakening and contraction, and $\preceq$ is a
pre-order among the elements of $I$. We shall use $a,b,\ldots$
 to range over elements in $I$ and we will assume that $\preceq$
is upwardly closed with respect to $U$, \ie, if $a \in U$ and $a \preceq
b$, then $b \in U$. 

The system $\sell$ is constructed by adding all the rules for
the linear logic connectives  except those for the exponentials. 

The rules for subexponentials are 
dereliction and
promotion of the subexponential labeled with $a \in I$
\[
\infer[\nbang{a}]{\vdash\nquest{a_1} F_1, \ldots \nquest{a_n} F_n, \nbang{a} G}{\vdash\nquest{a_1} F_1, \ldots \nquest{a_n} F_n,  G}
\qquad
\infer[\nquest{a}]{\vdash\Gamma, \nquest{a} G}{\vdash\Gamma, G} 
\]
Here, the rule $\nbang{a}$ has the side condition
that $a\preceq {a_i}$ for all $\tsl{i}$. That is, one can only
introduce a $\nbang{a}$   if
all other formulas in the sequent are marked with indices that are
greater than or equal to $a$. Moreover, for all 
indices $a \in U$, we add the usual  rules of weakening and contraction to $\quest^a$. 

We can enhance  the expressiveness of SELL with the subexponential quantifiers $\forallLoc$ and $\existsLoc$ (\cite{NigamOlartePimentel:concur-13,OlartePimentelNigam:tcs-15})
given by the rules (omitting the subexponential signature)
\[
 \infer[\forallLoc]{\vdash\Gamma, \forallLoc \typeloc{l_x}{a}. G}
 {\vdash \Gamma , G[l_{e}/l_x]}
\qquad
 \infer[\existsLoc]{ \vdash\Gamma, \existsLoc \typeloc{l_x}{a}. G}
 {\vdash\Gamma , G[l/l_x]}
\]
where   $l_e$ is fresh. 
Intuitively, subexponential variables play a similar role as eigenvariables. 
The generic variable $\typeloc{l_x}{a}$ represents any subexponential,
constant or variable in the ideal of $a$. Hence $l_x$ can be substituted 
by any  subexponential $l$ of type $b$ (i.e., $l:b$) if 
$b\preceq a$. 
We call the resulting system $\sellU$.

As shown in \cite{NigamOlartePimentel:concur-13,OlartePimentelNigam:tcs-15}, \sellU\ admits a cut-free, complete focused proof system
(presented in  the appendix).
That will be the system used throughout this text.

 \begin{theorem}
$\sellU$ admits cut-elimination for any subexponential signature $\Sigma$.
\end{theorem} 

\section{Relative Expressiveness Power of HyLL and SELL} 
\label{sec:hyll-sell}
We observe that,  while linear logic has only seven logically distinct prefixes of bangs and question-marks, \sell\ allows for an
unbounded number of such prefixes, \eg, $\nbang{i}$, or
$\nbang{i}\nquest{j}$. 
Hence,  by using
different prefixes, 
we  allow for the  specification of  richer systems where subexponentials are used to mark different modalities/states. For instance,
subexponentials can be used to represent contexts of proof
systems~\cite{DBLP:journals/entcs/NigamPR11}; to 
specify systems with temporal, epistemic and spatial modalities \cite{NigamOlartePimentel:concur-13,OlartePimentelNigam:tcs-15} and soft-constraints or preferences \cite{TLP:9303136}; 
 to 
specify Bigraphs \cite{DBLP:conf/lpar/ChaudhuriR15}; and to specify and verify biological  \cite{DBLP:journals/entcs/ChiarugiFHO16} and multimedia interacting systems \cite{DBLP:conf/mcm2/AriasDOR15}. 

One may wonder whether the use of worlds in HyLL  increases also  the expressiveness of LL. In this section we prove that this is not the case by showing that HyLL rules can be directly encoded into LL by using the methods proposed in \cite{DBLP:journals/tcs/MillerP13}. 

\subsection{HyLL and Linear Logic}\label{sec:hyll-ll}
In~\cite{DBLP:journals/tcs/MillerP13} classical linear logic (LL) was used as the logical framework for specifying a number of logical and computational systems. 
The idea is  simple:
use
two meta-level predicates $\lF{\cdot}$ and $\rF{\cdot}$ for  identifying objects that appear on the left or on the right side of the sequents in the object logic, respectively. 
Hence, object-level sequents of the form
${B_1,\ldots,B_n}\vdash{C_1,\ldots,C_m}$ (where $n,m\ge0$) are specified
as  the multiset
$\lft{B_1},\ldots,\lft{B_n},\rght{C_1},\ldots,\rght{C_m}$. If an object-formula $B$ is in a (object-level) classical context, it will be specified in LL as  $\quest\lft{B}$ or $\quest\rght{B}$ (depending on the side of $B$ in the original sequent). Hence HyLL sequents of the form
$\Delta;\Gamma\vdash C$ will be encoded in LL as
$\quest\lft\Delta\lpar\lft{\Gamma}\lpar\rght{C}$
where, if $\Psi = \{F_1,...,F_n\}$, then  $\lft\Psi=\lft{F_1} \lpar ... \lpar \lft{F_n}$ and 
$\quest\lft\Psi= \quest\lft{F_1} \lpar ... \lpar \quest\lft{F_n}$ (similarly for $\rF{\cdot}$). 
%

Inference rules are specified by a rewriting clause that replaces
the active formula in the conclusion by the active formulas in the premises. The  
 linear logic connectives  indicate how these
 object level 
 formulas are connected: contexts are copied ($\with$) or split ($\otimes$), in different inference rules ($\oplus$) or in the same sequent ($\lpar$).
As a matter of example, 
 the additive version of the inference rules 
for conjunction in classical logic
\[
\infer[\wedge_{L1}]{\Delta, A \wedge B \lra \Gamma}{\Delta, A \lra \Gamma}
\quad
\infer[\wedge_{L2}]{\Delta, A \wedge B \lra \Gamma}{\Delta, B \lra \Gamma}
\quad
\infer[\wedge_R]{\Delta \lra \Gamma, A \wedge B}{
 \deduce{\Delta \lra \Gamma,A}{}
 &
\deduce{\Delta \lra \Gamma,B}{}
}
\]
 can be specified as 
 \[
 \begin{array}{lll}
 \wedge_{L}: \exists A,B. (\lF{A\wedge B}^\perp \otimes (\lF{A} \oplus \lF{B} ))
 \qquad
  \wedge_{R} : \exists A,B. (\rF{A\wedge B}^\perp \otimes (\rF{A} \with \rF{B} ))\\
 \end{array}
 \]

%

The following definition shows how to encode HyLL inference rules into LL. 

\begin{definition}[HyLL rules into LL]\label{def:enc}
Let $\wtype$, $\fotype$, $\htype$ and $\ltype$ denote, respectively,  the types for worlds, (first-order) objects,  HyLL judgments  and LL formulas. Let $\rF{\cdot}$  and  $\lF{\cdot}$  be  predicates of the type $\htype \to \ltype$ and   $A$,  $B$,  $C$ have, respectively, types  $\wtype \to \htype$, $\fotype \to \htype$ and $\htype$. 
The encoding of HyLL inference rules into LL is depicted in Figure \ref{fig:hyllLL} (we  omit the encoding of most of the linear logic connectives that can be found in \cite{DBLP:journals/tcs/MillerP13}). 

\begin{figure}
\resizebox{\textwidth}{!}{
$
\begin{array}{lllllll}
 \otimes~R&:& \exists C,C',w. (\rF{ (C \otimes C')@w}^\perp \otimes  \rF{C@w} \otimes   \rF{C'@w})
 &\quad&
 \otimes~L & :& \exists C,C',w. (\lF{ (C \otimes C')@w}^\perp \otimes ( \lF{C@w} ~\invamp~  \lF{C'@w}))\\

\at~R&:&  \exists C, u, w. (\rF{ (C~\at~u)@w}^\perp \otimes \rF{C@u})
& \quad &
\at~L &:&   \exists C, u, w. (\lF{ (C~\at~u)@w}^\perp \otimes \lF{C@u}) \\

\downarrow R &:&   \exists A, u, w. (\rF{ \downarrow u. A@w}^\perp \otimes \rF{(A~w)@w})
& \quad &
\downarrow L &:&   \exists A, u, w. (\lF{ \downarrow u. A@w}^\perp \otimes \lF{(A~w)@w})\\

\forall R  (F)&:&   \exists B, u. (\rF{ \forall x. B@u}^\perp \otimes \forall x. \rF{(B~x)@u})
& \quad &
\forall L  (F)&:&  \exists B, u. (\lF{ \forall x. B@u}^\perp \otimes \exists x. \lF{(B~x)@u})
\\
\forall R(W)&:& \exists A, u. (\rF{ \forall v. A@u}^\perp \otimes \forall v. \rF{(A~v)@u})
& \quad &
\forall L(W)&:&  \exists A, u. (\lF{ \forall v. A@u}^\perp \otimes \exists v. \lF{(A~v)@u})
\\
\bang_L &:& \exists C,w. (\lF{!C@w}^\perp \otimes \quest \lF{C@w})
& \quad &
Init&:& \exists C,w. (\lF{C@w}^\perp \otimes \rF{C@w}^\perp)
\end{array}
$
}
\caption{HyLL rules into LL. (Definition \ref{def:enc}) \label{fig:hyllLL}}
\end{figure}
\end{definition}

Observe that  left and right inference rules for the hybrid connectives  ($\at~$ and $\downarrow$) are the same (see Section \ref{sec:hyll}). This is reflected in the duality of the encoding where we only replace $\rF{\cdot}$ with $\lF{\cdot}$. 
 Observe also that
the inference rules for the quantifiers (first-order  and worlds) look
the same. The difference is on the type of the variables involved. Since $A$ has type $\wtype \to \htype$, the encoding clause $\forall R(W)$
guarantees that the variable $v$ has type $\wtype$. Analogously, since  $B$ has type $\fotype \to \htype$, then  $x$  has type $\fotype$ in the
 clause $\forall R(F)$.
This neat way of controlling the behavior of objects by using types is also inherited by the encoding of the other object level inference rules.

The following theorem shows that, in fact, the encoding of HyLL into LL
is adequate in the sense that a focused step in FLL corresponds {\em exactly} to the application of one inference rule in HyLL.

\begin{theorem}[Adequacy]\label{th:hyll-ll}
Let $\Upsilon$ be the set of clauses 
in Figure \ref{fig:hyllLL}. 
The sequent $\Gamma;\Delta \vdash F@w$ is provable in HyLL iff~ 
$\Up{\Upsilon, \lF{\Gamma}}{\lF{\Delta},\rF{F@w}}{\cdot}$
is provable in FLL. Moreover, the adequacy of the encodings is  on the {\em level of derivations} meaning that, when focusing on a 
specification clause, the bipole derivation  corresponds exactly to applying the introduction rule at the object level.
\end{theorem}
\begin{proof} We will illustrate here the case for rule $\at_L$, the other cases are similar. 
Applying the object level rule
\[
\infer[\at_L]{\Gamma;\Delta, (A ~\at~ u)@w \vdash C @v}{
  \Gamma;\Delta, A@u \vdash C @v
}
\]
corresponds to  deciding on the LL formula given by the encoding of the rule $\at_L$  (stored in   $\Upsilon$). Due to focusing, the derivation in LL has necessarily the shape\\

\resizebox{.95\textwidth}{!}{
$
\infer[D_2]{
   \Up{\Upsilon, \lF{\Gamma}}{\lF{\Delta},\lF{(A ~\at~u)@w} , \rF{C@v}}{\cdot}
}
{
 \infer[3 \times \exists]{ \Down{\Upsilon, \lF{\Gamma}}{\lF{\Delta},\lF{(A ~\at~u)@w} , \rF{C@v}}{ \exists C, u, w. (\lF{ (C~\at~u)@w}^\perp \otimes \lF{C@u})}}{
  \infer[\otimes]{\Down{ \Upsilon, \lF{\Gamma}}{\lF{\Delta},\lF{(A ~\at~u)@w} , \rF{C@v}}{ \lF{ (A~\at~u)@w}^\perp \otimes \lF{A@u}}}{
   \infer[I_1]{\Down{ \Upsilon, \lF{\Gamma}}{ \lF{(A ~\at~u)@w}}{ \lF{ (A~\at~u)@w}^\perp}}{}
   &
   \infer[R\Downarrow,R\Uparrow]{ \Down{ \Upsilon, \lF{\Gamma}}{\lF{\Delta}, \rF{C@v}}{\lF{A@u}}}{
    \deduce{ \Up{ \Upsilon, \lF{\Gamma}}{\lF{\Delta}, \rF{C@v},\lF{A@u}}{\cdot}}{}
   }
  }
 }
}
$}
\ \\
That is, the LL formula corresponding to $(A ~\at~u)@w$ is consumed and,  in the end of the focused phase, the encoding of  $A@u$ is stored into the linear context. 
This mimics exactly the application of the Rule $\at_L$ in HyLL. 
\end{proof}

\subsection{HyLL and SELL}\label{sec:hyllsell}
Linear logic allows for the specification of two kinds of context maintenance: both weakening and contraction are available (classical context)  or  neither is available (linear context). 
That is,  when we encode (linear) judgments in HyLL belonging to different worlds, the resulting meta-level atomic formulas will be stored in the same (linear) LL context. The same happens with classical  HyLL judgments  and the classical LL context.

Although this is perfectly fine, 
 encoding  HyLL into \sellU\  allows for a better understanding  of worlds in HyLL
 as we shall see. We use subexponentials to represent worlds, where each world has its own  linear context. More precisely, a HyLL judgment of the shape $F@w$ in the (left)
linear context is encoded as the \sellU\ formula $\nquest{w}\lF{F@w}$. 
Hence, HyLL judgments that hold at world $w$ are stored at the $w$ linear context of \sellU. A judgment of the form $G@w$ in the
classical HyLL context is encoded as the \sellU\ formula $\nquest{\copysell}\nquest{w}\lF{G@w}$. That is, the encoding of $G@w$ is stored in the unbounded (classical) subexponential context $\copysell$.

The next definition introduces the encoding of  HyLL inference rules   into \sellU. Observe that, surprisingly, the subexponential structure needed is flat and it does not reflect the order on worlds. This is explained by the fact that worlds in HyLL do not control the context on rules as the promotion rule in SELL does. This also explains why HyLL does not add any expressive power to LL.

%

%

\begin{definition}\label{def:enc:sell}
Let $\wtype,\fotype,\htype,\rF{\cdot},\lF{\cdot},A,B,C$ be as in Definition~\ref{def:enc} and $\ltype$ be the type for \sellU formulas. 
Given a HyLL constraint domain $\cal W$, consider a subexponential signature 
$\Sigma = \tup{I, \preceq, U}$ 
such that 
$I = \cal W \cup \{\infty,\copysell\}$, $w \preceq \infty$ for any $w\in I$ and, for any $u,w \in \cal W\cup\{\copysell\}$, $u \not\preceq w$.  Moreover, $U = \{\copysell,\infty\}$. The encoding of HyLL inference rules into \sellU\ is depicted in Figure \ref{fig:hyllSELL} (we  omit the encodings of the other connectives, that follow similarly).
\begin{figure}
\resizebox{.85\textwidth}{!}{
$
\begin{array}{lllllll}
\otimes~R &:& \exists C, C'. \existsLoc \typeloc{w}{\infty}. (\nbang{w} \rF{ (C \otimes C')@w}^\perp ~\otimes~ \nquest{w} \rF{C@w} ~\otimes~ \nquest{w} \rF{C'@w})\\

\at~R &:  &\exists A. \existsLoc \typeloc{u}{\infty}, \typeloc{w}{\infty}. (\nbang{w} \rF{ (A~\at~u)@w}^\perp ~\otimes~ ?^u \rF{A@u})\\

\at~L &:  &\exists A. \existsLoc\typeloc{u}{\infty}, \typeloc{w}{\infty}. (\nbang{w} \lF{ (A~\at~u)@w}^\perp ~\otimes~ \nquest{u} \lF{A@u})\\
\downarrow R&:  &\exists A. \existsLoc \typeloc{u}{\infty}, \typeloc{w}{\infty}. (\nbang{w} \rF{ \downarrow u. A@w}^\perp ~\otimes~ \nquest{w} \rF{(A~w)@w})
\\
\downarrow L&:  &\exists A. \existsLoc \typeloc{u}{\infty}, \typeloc{w}{\infty}. (\nbang{w} \lF{ \downarrow u. A@w}^\perp ~\otimes~ \nquest{w} \lF{(A~w)@w})\\
\forall R(F)&:& \exists A, \existsLoc \typeloc{w}{\infty}. (\nbang{w}\rF{ \forall x. B@w}^\perp \otimes \forall x.  \nquest{w}\rF{(B~x)@w})
\\
\forall R(W)&:& \exists A, \existsLoc \typeloc{w}{\infty}. (\nbang{w}\rF{ \forall v. A@w}^\perp \otimes \forallLoc \typeloc{v}{\infty}. \nquest{w}\rF{(A~v)@w})
\\
\bang_L &:& \exists C,\existsLoc \typeloc{w}{\infty}. (\nbang{w} \lF{!C@w}^\perp \otimes \nquest {\copysell}\nquest {w} \lF{C@w})
\end{array}
$
}
\caption{HyLL rules into \sellU. (Definition \ref{def:enc:sell}) \label{fig:hyllSELL}}
\end{figure}
\end{definition}


Note that  $\typeloc{w}{\infty}$ represents \emph{any subexponential}  in the ideal of $\infty$. This means that, in   $\existsLoc\typeloc{w}{\infty}.F$, the subexponential variable $w$ could be substituted, in principle, by 
{\em any} element of $I$. But note that, since world symbols are restricted to $\cal W$, substituting $w$ by $\copysell$ or $\infty$
would not match any encoded formula in the context. That is, the proposed subexponential signature correctly specifies the role of worlds in HyLL.

The following theorem shows that our encoding is indeed adequate. 

\begin{theorem}[Adequacy]\label{th:hyll-sell}
Let $\Upsilon$ be the set of formulas resulting from the encoding 
in Definition \ref{def:enc:sell}.
The sequent $\Gamma;\Delta \vdash F@w$ is provable in HyLL iff~ 
$\Up{
\copysell: \{\Upsilon, \lF{\Gamma}\},w_i:\lF{\Delta},  \nquest{w}\rF{F@w}  }{\cdot}$
is provable in \sellU.\footnote{Clarifying some notation:  if $\Delta=\{F_1@w_1,\ldots,F_n@w_n\}$, then $\nquest{w_i}\lF{\Delta}=\nquest{w_1}\lF{F_1@w_1},\ldots,\nquest{w_n}\lF{F_n@w_n}$. Observe that, in the negative phase, such formulas will be stored at their respective contexts, that will be  represented by $w_i:\lF{\Delta}$.} Moreover, the adequacy of the encodings is on the {\em level of derivations}.
\end{theorem}
\begin{proof} Again, we will 
 consider 
the rule $\at_L$, as the other cases are similar. 
If we decide to focus on the \sellU\ formula corresponding to the
encoding of $\at_L$ (stored in $\nquest{\copysell} \Upsilon$), we
obtain  \\

\resizebox{\textwidth}{!}{
$
\hskip -1cm
\infer[D]{
\Up{
\copysell: \{\Upsilon, \lF{\Gamma}\},w_i:\lF{\Delta}, w:\lF{(A ~\at~u)@w} , v:\rF{C@v}  }{\cdot}
}{
 \infer[\exists,\existsLoc]{ \Down{\copysell:\{ \Upsilon, \lF{\Gamma}\},w_i:\lF{\Delta},w:\lF{(A ~\at~u)@w} , v:\rF{C@v}}{\cdot}{\exists C, \existsLoc u, w. (\nbang{w}\lF{ (C~\at~u)@w}^\perp \nquest{u}\otimes \lF{C@u})}}{
  \infer[\otimes]{ \Down{\copysell:\{ \Upsilon, \lF{\Gamma}\},w_i:\lF{\Delta},w:\lF{(A ~\at~u)@w} , v:\rF{C@v}}{\cdot}{ \nbang{w}\lF{ (A~\at~u)@w}^\perp \otimes \nquest{u}\lF{A@u}}}{
   \infer[\nbang{s}]{ \Down{\copysell: \{\Upsilon, \lF{\Gamma}\}, w:\lF{(A ~\at~u)@w}}{\cdot} {\nbang{w}\lF{ (A~\at~u)@w}^\perp}}{
     \infer[D,I]{\Up{w:\lF{(A ~\at~u)@w}}{\cdot}{ \lF{ (A~\at~u)@w}^\perp}}{}
   }
   &
   \infer[R\Uparrow,\nquest{s}{}]{ \Down{\copysell:\{ \Upsilon, \lF{\Gamma}\},w_i:\lF{\Delta}, v:\rF{C@v}}{\cdot}{\nquest{u}\lF{A@u}}}{
    \Up{\copysell:\{ \Upsilon, \lF{\Gamma}\},w_i:\lF{\Delta}, v:\rF{C@v},u:\lF{A@u}}{\cdot}{\cdot}
   }
  }
 }
}
$}
\ \\

Observe that,   in a (focused) derivation proving  $\nbang{w} F$, 
the only contexts that can be present are  $w$ and the $\infty$ contexts due to the promotion rule and the ordering in $\Sigma$. Since the encoding does not store any formula into the context $\infty$,  the formula  $\nbang{w} F$ must necessarily be  proved from  the formulas stored in $w$. Thus, unlike the LL derivation after Theorem  \ref{th:hyll-ll},  the context $\copysell$ is weakened
in the left-hand side derivation since $\copysell \not\preceq w$. Hence $\lF{(A ~\at~u)@w}$ stored initially in the location $w$ is substituted by $\lF{A@u}$ in the location $u$ in one focused step.
\end{proof}

\subsection{Information Confinement} 
One of the features needed to specify spatial modalities  is information \emph{confinement}: a space/world can be inconsistent and this does not imply the inconsistency of the whole system. 
We finish this section by showing that information confinement, a feature that can be specified in SELL,   cannot be modeled in HyLL. 

In   \cite{NigamOlartePimentel:concur-13}  the combination of subexponentials of the form $\nbang{w}\nquest{w}$ was used in order to specify information confinement in \sell. 
  More precisely, since the sequents (in a 2-sided sequent presentation) 
$\nbang{w}\nquest{w}\zero \vdash \zero$ 
and $\nbang{w}\nquest{w}\zero \vdash \nbang{v}\nquest{v}\zero$  are {\em not} provable in SELL, it is possible to specify systems  where inconsistency is local to a given space and does not propagate to the other locations.

In HyLL, however, it is not possible to confine inconsistency: the HyLL rule 
\[
\infer[\mathbf{0} L]{\Gamma;\Delta, \zero@u \vdash F@w}{}
\]
shows that {\em any} formula $F$ in {\em any} world $w$ is derivable from $\mathbf{0}$ appearing in
{\em any} world $u$. 
Observe that, even if we exchange the rule $\mathbf{0} L$ for a weaker version
\[
\infer[0'_L]{\Gamma;\Delta, \zero@w \vdash F@w}{}
\]
the rule $\mathbf{0} L$ would still be admissible
\[
\infer[\cut]{\Gamma;\Delta,\zero@w \vdash F@v}{
  \infer[0'_L]{\Gamma;\Delta,\zero@w \vdash (\zero ~\at~ v)@w}{}
  &
  \infer[\at_L]{\Gamma;\Delta,(\zero ~\at~ v)@w\vdash  F@v}{
   \infer[0'_L]{\Gamma;\Delta,\zero@v \vdash F@v}{}
  }
 }
\]
 
\section{Computation Tree Logic (CTL) in Linear Logic.}\label{sec:temporal}
%

Hybrid linear logic is expressive enough to encode 
some forms of modal operators, thus allowing for the specification of properties of transition systems. As mentioned in \cite{deMaria-Despeyroux-Felty:14-fmmb}, it is possible to encode 
CTL temporal operators into HyLL considering 
 existential ($\tE$)  and bounded universal ($\tA$) path quantifiers. 
We show in this section the limitation of such encodings and how to  fully capture $\tE$ and $\tA$  CTL quantifiers in linear logic with fixed points. For that, we shall use the system $\mu$MALL \cite{DBLP:journals/tocl/Baelde12} that extends MALL
(multiplicative, additive linear logic)  with  fixed point operators.

\paragraph{CTL connectives and path quantifiers}
Let us recall the meaning of the temporal operators in CTL.
$\tX$ (Next)  means ``at the next state''.
$\tF$ (Future) means ``in some future'' while $\tG$ (Globally) means ``in all futures''.
  $\varphi \tU \psi$ ($\varphi$ until $\psi$) means 
``from now, $\varphi$ will be true in every steps until some future point 
(possibly including now)
where $\psi$ holds (and from that point on, $\varphi$ can be true or false)''.

The CTL quantifier $\tE$ (Exists) means ``for some path'' while $\tA$ (All) 
means ``for all paths''. Formulas in CTL are built from propositional variables $a,b,c,...$, the usual propositional logic connectives  and the temporal connectives preceded by a path quantifier: 
\begin{equation}
\begin{array}{lll}
F &::=&  p \mid F \wedge F \mid F \vee F 
\mid 
 {\bf Q} \tX F
\mid {\bf Q} \tF F
\mid {\bf Q} \tG F
\mid {\bf Q} [F \tU F  ]
\qquad \bf Q \in \{\tA,\tE\}
\end{array}
\label{eq:synCTL}
\end{equation}
where $p$ is a state formula. 
\paragraph{Transition Systems} 
Consider a set of propositional CTL variables   $\cn{V}=\{a_1,...,a_n\}$.  A state $\stateS$  is  a 
valuation from $\cn{V}$ into the set $\{\true,\false\}$. 
We shall use $\presS{a_i}$  (resp. $\absS{a_i}$) 
to denote that $\stateS(a_i)=\true$  
(resp. $\stateS(a_i)=\false$). Hence, a state $\stateS$ on the set $\texttt{V}$ can be seen as a conjunction of the form $\texttt{p}_1(a_1) \wedge ... \wedge \texttt{p}_n(a_n)$ where $\texttt{p}_i \in \{\presSS,\absSS\}$. 

We  consider transition systems defined by states as above and  transition rules of the form  $r:\stateS\to \stateS'$.
For instance, if $\texttt{V}=\{a,b\}$, the transition rule $r: \presS{a} \wedge \absS{b} \to \absS{a} \wedge \presS{b}$  enables a transition from a state  $\stateS=\{a\mapsto\true,b\mapsto\false\}$  to the state $\stateS' =\{a\mapsto\false, b\mapsto\true\}$. We shall use   $\stateS\rede{r}\stateS'$ to denote such transitions.  

\subsection{Transition Systems and HyLL} \label{sec:tsHyLL}
In order to specify reachability properties in 
transition systems, some modal connectives are defined  in HyLL 
\cite{ChaudhuriDespeyroux:14}:
\[
\begin{array}{cclcccl}
    \Box A &\eqdef& {\downarrow} u.~ \forall w.~ (A ~\at~ u . w) &\qquad&
    \Diamond A &\eqdef& {\downarrow} u.~ \exists w.~ (A ~\at~ u . w) \\
     \delay{v} A &\eqdef& {\downarrow} u.~ (A ~\at~ u . v) &\qquad&
     A ~\tU~ B &\eqdef& \downarrow u.~ \exists v. \left(B ~\at~ u.v  
     ~~\with~~ \forall w \prec v.~ A ~\at~ u.w \right)
     \end{array}
\]
$\Box A$ (resp. $\Diamond A$) represents all (resp. some) state(s) satisfying $A$ and reachable in some path from now. 
The connective $\delta$ represents a form of delay:  $\delay{v} A$   stands for an \emph{intermediate state} in a
transition to $A$. Informally it can be thought to be ``$v$ before $A$''. 
$A \tU B$ represents that $A$ holds in all the steps until $B$ holds. 


We may use such modal operators in order to encode some features of transition systems as HyLL formulas as follows.
Consider a set $\texttt{V}=\{a_1,...,a_n\}$ of  propositional variables, let $\stateS=\texttt{p}_1(a_1) \wedge \cdots \wedge\texttt{p}_n(a_n)$  represent a state where $\texttt{p}_i\in \{\presSS,\absSS\}$ and $r:\stateS\to \stateS'$ be a rule specifying a state transition. We define the encoding 
$\os \cdot\cs$ from CTL states and state transitions to HyLL as
\[\begin{array}{lcl}
\os \presS{a_i}\cs  = \presS{a_i}
&\qquad\qquad&
\os \absS{a_i}\cs = \absS{a_i}\\
\os \stateS\cs =\bigotimes\limits_{i\in 1..n} \os \texttt{p}_i(a_i)\cs
 & &
\os r:\stateS \to\stateS' \cs=
\forall w.\left( (\os\stateS\cs ~\at~ w  )
\limp \delay{1} (\os \stateS' \cs) ~\at~ w\right)
\end{array}
\]
Moreover,  let $F,G$ be CTL formulas built from 
states and the connectives
$\wedge,\vee, \tU,\tE\tX,\tE\tF$. We can define $\encCTL{F}$ as  
\[
\begin{array}{lll l lll}
\encCTL{\stateS} & = & \os \stateS\cs
&\quad&
\encCTL{ F \wedge G} & = & \encCTL{F} \with \encCTL{G}
\\
\encCTL{ F \vee G} & = & \encCTL{F} \oplus \encCTL{G}
&\quad&
\encCTL{\tE[F \tU G]} & =& \encCTL{F} \tU~ \encCTL{G}
\\
\encCTL{\tE\tX F} & =& \delay{1}\encCTL{F}
&\quad&
\encCTL{\tE\tF F}& =& \Diamond\encCTL{F} \\
\end{array}
\]

It is easy to see that such encodings are {\em faithful}, that is, a 
 (CTL)  formula $F$ holds at state $\stateS$ in the system defined by
 the transition rules $\Rscr$ if and only if the sequent $\os
 \Rscr\cs@0; \os\stateS\cs @ w\vdash \encCTL{ F} @ w$ is provable in HyLL 
(see the appendix).
In fact, since the left linear context is always constituted by atoms, the only action that can be performed is to apply transition
 rules up to reaching the state satisfying $F$, which is reachable in a finite number of steps for this  CTL's limited grammar.


However, 
the above encodings cannot be extended to consider formulas of the shape   $\tE\tG F$. In fact, the natural choice  would be 
$\encCTL{ \tE \tG F}= \Box \encCTL{F}$,
but this encoding would not be adequate. Consider, for instance, a system with only one rule $\Rscr= \{ r: \stateS \to \stateS\}$ that loops on the same state. Clearly, in CTL, $\stateS$ satisfies the formula $\tE\tG \stateS$. Now, consider the HyLL sequent 
$\os \Rscr\cs @ 0; \os\stateS\cs @ w\vdash \Box\encCTL{ \stateS}@ w$. If we decide to introduce the connectives on the right, we obtain a derivation of the shape 
\[
\infer=[\downarrow_R,\forall_R,\at_R]{\os \Rscr\cs@0; \os\stateS\cs @ w\vdash \Box\encCTL{ \stateS} @ w}{\os \Rscr\cs@0; \os\stateS\cs @ w\vdash \os \stateS \cs@w.v}
\]
where $v$ is fresh.
Furthermore, if we use the implication in $\os \Rscr\cs@0$, we obtain a derivation of the shape:
\[
\infer=[copy, \forall_L, \limp_L]{\os \Rscr\cs@0; \os\stateS\cs @ w\vdash G}
{\os \Rscr\cs@0; \os\stateS\cs @ (w+1)\vdash G}
\]

Therefore,  the left and right states in the sequent $\os \Rscr\cs@0; \os\stateS\cs @ (k+n)\vdash \os \stateS \cs_v@w.v$ will never match, and this sequent is not provable. Saying this in other way, the resources in the context  are enough for proving the property for a (bounded) $n$ but not for all natural numbers. For proving this,  one {\em necessarily} needs (meta-level) induction which is the same as using fixed points. The next section shows how to do that with linear logic with fixed point operators. 


%


\subsection{Encoding  $\tE$ and $\tA$ quantifiers in linear logic with fixed points}
In order to prove (in CTL) the formula  $\tA \tF F$ at state $\stateS$, we have to check if $\stateS$ satisfies $F$. If this is not the case, we have to check that 
$\tA \tF F$ holds for all the  successors of $\stateS$    (i.e., for all $\stateS'$ s.t. $\stateS \rede{r} \stateS'$ for some transition rule $r$). Hence, the definition of $\tA\tF$ is recursive and it is usually characterized as a (least) fixed point. 

One way to capture this behavior would by adding fixed point operators  to HyLL. 
But it is simpler to rely on existing systems for linear logic with fixed points. In the following, we show that it is possible to characterize the CTL formulas built from the syntax \eqref{eq:synCTL} into the system $\mu$MALL~\cite{DBLP:journals/tocl/Baelde12} that adds to linear logic (without exponentials) least and greatest fixed points. 

$\mu$MALL shares with linear logic all the proof rules for the additive and multiplicative connectives plus the following two rules\footnote{ $\mu$MALL also consider rules for equality and inequality but we do not need them in our developments.}
\[
\begin{array}{c}
\infer[\nu]{
 \vdash \Delta, \nu B \vec{t}
}{
 \vdash \Delta, S\vec{t}
&\quad
\vec{x} \vdash B~S\vec{x}, (S\vec{x})^\perp
}
\quad 
\infer[\mu]{
 \vdash \Delta, \mu B \vec{t}
}{
 \vdash \Delta, B(\mu B )\vec{t}
}
\end{array}
\]
where  $S$ is the (co)inductive invariant. The $\mu$ rule corresponds to unfolding while $\nu$  allows for (co)induction. 

\paragraph{Path quantifiers as fixpoints}
The usual interpretation of the   CTL quantifiers as fixed points  (see e.g., \cite{DBLP:journals/iandc/BurchCMDH92}) is
\[
\begin{array}{lll  lll   lll}
\tE \tF F  &=& \mu Y. F \vee \tE \tX Y
&\quad
\tE \tG F  &=& \nu Y. F \wedge \tE \tX Y
&\quad
\tE [F~\tU~G]  &=& \mu Y. G \vee (F \wedge \tE\tX Y)
\\
\tA\tF F & = & \mu  Y. F \vee \tA \tX Y
&\quad
\tA\tG F & = & \nu  Y. F \wedge \tA \tX Y
&\quad
\tA [F~\tU~G]  &=& \mu Y. G \vee (F \wedge \tA\tX Y)
\end{array}
\]

In CTL, the considered transition system is assumed to be serial, i.e. every state has at least one successor. 
This means that, in every state, there is at least one fireable rule. 

The next definition shows how to encode CTL formulas into $\mu$MALL. 
 
\begin{definition}[CTL into $\mu$MALL]\label{def:aq} Let $\Rscr$ be a set of transition rules. The encoding of ${\bf Q} \tX$, 
${\bf Q} \tF$ and ${\bf Q} \tG$, for $\bf Q \in \{\tA,\tE\}$ is  in Figure \ref{fig:ctl_fp}. 
Given a state $\stateS=\texttt{p}_1(a_1) \wedge \cdots \wedge\texttt{p}_n(a_n)  $ (as in Section \ref{sec:tsHyLL}), we define
\[
\begin{array}{lll l lll ll}
\os  \presS{a_i}\cs&=&a_i  
& &
\os  \absS{a_i}\cs &=& a_i^\perp
\\
\os s\cs&=& \os\texttt{p}_1(a_1)\cs^\perp  \invamp~ \cdots  \invamp ~\os\texttt{p}_n(a_n)\cs^\perp
& &
\os p\cs &=& \positiveS(p)
\\
\positiveS(\stateS) &=& \os\texttt{p}_1(a_1) \cs\otimes~ \cdots   \otimes \os\texttt{p}_n(a_n)\cs = \os \stateS\cs^\perp
\\
\negativeS(\stateS) &=& (\os\texttt{p}_1(a_1)\cs^\perp\otimes \top) \oplus  \cdots  \oplus (\os\texttt{p}_n(a_n)\cs^\perp \otimes \top)
\end{array}
\]
where $p$ is a state formula.\footnote{
It is useful to allow the state property to mention only a subset of  the propositional variables in $\cn{V}$. In that case, we can define 
 $\os\texttt{p}_i(a_i)\cs $ as above if
$a_i$ occurs in $p$ and $\top$ otherwise.} Finally, we map the CTL connectives  $\wedge$ and $\vee$   into $\with$ and $\oplus$, respectively. 

\begin{figure}
\resizebox{.9\textwidth}{!}{
$
\begin{array}{lll}
\encCTLR{\tA \tX F} &=&    \bigwith\limits_{\stateS\to \stateS'\in \Rscr } \left(  {\negativeS(\stateS)} \oplus ( {\positiveS(\stateS) } \otimes \left( { \os \stateS' \cs } \invamp~ \phi\right)  \right)
\\
\encCTLR{\tE \tX F} &=&    \bigoplus\limits_{\stateS\to \stateS'\in \Rscr } \left(  {\positiveS( \stateS) } \otimes \left( { \os \stateS' \cs } \invamp~ \phi\right)  \right)
\\
\encCTLR{\tA \tF F} &=& \mu Y.  ~ \phi  \oplus    \bigwith\limits_{\stateS\to \stateS' \in \Rscr} \left(  {\negativeS(\stateS) } \oplus ( {\positiveS( \stateS) } \otimes \left( { \os \stateS' \cs } \invamp~ Y\right)  \right)
\\
\encCTLR{\tE \tF F} &=& \mu Y.  ~ \phi  \oplus    \bigoplus\limits_{\stateS\to \stateS' \in \Rscr} \left(    {\positiveS( \stateS) } \otimes \left( \os \stateS' \cs \invamp~ Y\right)  \right)
\\
\encCTLR{\tA \tG F} &=& \nu Y. ~   \phi \with  \bigwith\limits_{\stateS\to \stateS' \in \Rscr} \left(  \negativeS( \stateS ) \oplus ( {\positiveS( \stateS)} \otimes \left( { \os \stateS' \cs } \invamp~ Y \right)  \right) 
\\
\encCTLR{\tE \tG F} &=& \nu Y. ~   \phi \with  \bigoplus\limits_{\stateS\to \stateS' \in \Rscr} \left(   {\positiveS( \stateS)} \otimes \left( \os \stateS' \cs \invamp~ Y \right)  \right) 
\\
\encCTLR{\tA [F ~\tU~ G]} &=& \mu Y.  \psi \oplus \left(  \phi \with \bigwith\limits_{\stateS\to \stateS'\in \Rscr} \left(  \negativeS( \stateS ) \oplus ( \positiveS( \stateS) \otimes \left( { \os \stateS' \cs }  \invamp~  Y  \right))  \right)\right) 
\\
\encCTLR{\tE [F ~\tU~ G]} &=& \mu Y.  \psi \oplus \left(  \phi \with \bigoplus\limits_{\stateS\to \stateS'\in \Rscr} \left(   \positiveS( \stateS) \otimes \left( { \os \stateS' \cs} \invamp~  Y  \right)  \right)\right) 
\end{array}
$
}
\caption{Encoding of CTL temporal operators into $\mu$MALL. Here, $\phi= \encCTLR{F}{}$
and $\psi = \encCTLR{G}{}$.
\label{fig:ctl_fp}}
\end{figure}
\end{definition}

Let us give some intuition.  Consider the rule $r: \stateS \to \stateS'$. 
The formula $\positiveS(\stateS)$ (resp. $\negativeS(\stateS)$) 
tests if $r$ can (resp. cannot) be fired at the current state. 
The encoding of the temporal quantifiers relies on the following principles. For each transition rule, we test if the rule can be fired or not. If it can be fired, then the current state is transformed into the new state. 
The encoding of 
$\tA$ (resp. $\tE$)
test all (resp. one) of the fireable rules. This explains the use of $\bigwith$ (resp. $\bigoplus$).

\begin{example}\label{ex:AF}
Consider the temporal formula  $\tA\tF F$. We first check if $F$ holds in the current state. If this is not the case, 
 for each of the fireable rules, we  consume 
$\os \stateS \cs$  (using $\positiveS(\stateS)$)
and release $\os \stateS'\cs$, thus updating the current state. 
 For instance, 
consider the sequent $\vdash {\os\stateS\cs} ,  \encCTLR{\tA\tF F}$ and assume that 
the formula  $F$
 does not hold at state $\stateS$. If we decide to focus on $\encCTLR{\tA\tF F}$ we obtain a derivation of the shape\\

\resizebox{\textwidth}{!}{
$ 
\infer[\mu,\oplus,\with \quad~]{\vdash \os \stateS\cs , \mu B }{
 \vdash  \os \stateS\cs , \negativeS(\stateS_1) \oplus (\positiveS(\stateS_1) \otimes ( \os \stateS_1'\cs  \invamp~ \mu B)
  \quad ...\quad
  \vdash \os \stateS\cs , \negativeS(\stateS_m) \oplus (\positiveS(\stateS_m) \otimes ( \os \stateS_m'\cs  \invamp~ \mu B)
}
$
}
\ \\ The  premises correspond to proving if the rule $r_i$ is fireable or not. If $r_i:\stateS_i \to\stateS_i' $ is fireable, we observe a derivation of the shape:
\[
\infer[\oplus]{ \vdash \os \stateS\cs, \negativeS(\stateS_i)  \oplus (\positiveS(\stateS_i)   \otimes (\os \stateS_i'\cs\invamp~ \mu B) )}{
 \infer[\otimes,\invamp]{\vdash \os \stateS \cs, \positiveS(\stateS_i)   \otimes (\os \stateS_i'\cs\invamp~ \mu B) }{
  \deduce{\vdash \os \stateS_i'\cs, \mu B }{}
 }
 }
\]
where $\stateS$ becomes $\stateS_i'$ and, from that state, $\mu B$ must be proved. 
\end{example}
\begin{remark} Observe that, in all the clauses in Figure~\ref{fig:ctl_fp}, the formula 
$ { \positiveS(\stateS)} \otimes ( \os \stateS'\cs \invamp~ B$), 
is present. We could have written instead 
$ {\os r \cs} \limp B$, which reads closer to what we expect: ``assuming that   $r$ is fired, $B$ holds''. The formulas $(L\limp R)\limp B$ and $L\otimes (R\limp B)$ are not equivalent. In fact, the first formula is equivalent to $(L\otimes R^ \perp)\invamp B$ while the second is equivalent to $L\otimes (R^ \perp \invamp~ B)$. The first is stronger than the second in the sense that $B$ can choose the branch to move up with ($L$ or $R$), while the second forces $B$ to stick with $R$. Since the desired behavior is the second, offering an extra possibility is not good for proof search.
\end{remark}

\begin{theorem}\label{th:ctl_mumall}
Let $\texttt{V}=\{a_1,...,a_n\}$ be a set of propositional variables, $\Rscr$ be a set of transition rules on $\texttt{V}$, $F$ be a CTL formula and
$\ctlS{\stateS}{\Rscr} F$ denote that $F$ holds at state $\stateS$ in the transition system defined by $\Rscr$.
 Then, $\ctlS{\stateS}{\Rscr}  F$  iff
the sequent 
$\vdash\os \stateS\cs, \encCTLR{F} $ is provable in $\mu$MALL. 
\end{theorem}
\begin{myproof}[Proof sketch]
In the case of the least fixed point, the result comes easily since unfolding will always substitute $\stateS$ by a reachable  state $\stateS'$ (see derivation in  Example \ref{ex:AF}). For the greatest fixed point, we show that we can always provide the needed invariant. For example,
assume that the states in 
$S=\{s_1,...,s_n\}$ satisfy the formula $\tA\tG F$. We can show that   $I= \os s_1\cs^\perp \oplus \cdots \oplus \os s_n\cs^\perp$ is the greatest invariant  for proving the sequent $\vdash \os \stateS\cs, \encCTLR{\tA\tG F}$. See appendix for the detailed proof. 
\end{myproof}


Our encodings assume that each rule ``uses'' all the variables (either present or absent). 
This greatly simplifies the encodings and the adequacy proofs. 
In \cite{deMaria-Despeyroux-Felty:14-fmmb}, 
this restriction is not imposed, giving rise to  more compact rules.
%
%
%
Note that our restrictions on the use of variables are without loss of generality: if a rule does not use all the variables, we can preprocess the input and generate accordingly a set of rules satisfying our requirements. 

It it worth noticing that, in  Definition \ref{def:aq}, we do not encode the transition rules  as a theory (as we did in Section \ref{sec:tsHyLL}). The reason in the following.  On one hand,  the presence of a  formula $\os\Rscr \cs$ in the context, encoding the rules,   may allow us to move from the current state to a successor one. On the other hand, fixed points operators must be applied in order to go through paths, checking properties on them. These two actions should be coordinated, otherwise one would lose adequacy in the encodings. More precisely, by focusing on $\os \Rscr\cs$, we may ``jump'' a state without checking the needed property in that state. 
Therefore, the use of fixed points excludes the use of theories for encoding the transition system  and  we   must internalize the  transition rules in the definition of the path quantifiers  (see Definition~\ref{def:aq}). This  seems to be the accurate way of
controlling the use of rules in CTL. 

Finally,  note that the encoding of the CTL operators  does not use the exponentials.  That is, only the multiplicative/additive part of the logic  is  enough. 

\section{Concluding Remarks and Future Work} \label{sec:conclusion}
%
We compared the expressiveness, as logical frameworks, of two extensions of linear logic (LL). We show that it is possible to encode the logical rules of HyLL into LL. In order to better analyze the meaning of worlds in HyLL, we show that a flat subexponential structure 
suffices to encode HyLL into \sellU. We also show that information confinement cannot be specified in HyLL. Finally, with better insights about the meaning of HyLL's words, 
we pushed forward previous attempts of using HyLL to encode Computational Tree Logic (CTL). 
We showed that only by using meta-level induction (or fixed points inside the logic) it
is possible to faithfully encode CTL path quantifiers. 

There are some other logical frameworks that are extensions of LL, for example,  HLF\cite{reed06hylo}.
Being a logic in the LF family, HLF is based on natural deduction, hence
having a complex notion of ($\beta \eta$) normal forms
as well as lacking a focused 
system. 
Thus adequacy (of encodings of systems in HLF) results are often much harder to 
prove in HLF than in HyLL or in SELL.
%
  

While logical frameworks should be general enough for specifying and verifying properties of a large number of systems, some logical frameworks may be more suitable for dealing with specific applications than others. Hence, it makes little sense to search for ``the universal logical framework''. However, it is often salutary to establish connections between frameworks, specially when they are 
meant to reason about the same set of systems.

In this context, both HyLL and \sell\ have been used for formalizing and analyzing biological systems 
\cite{deMaria-Despeyroux-Felty:14-fmmb,DBLP:journals/entcs/ChiarugiFHO16}. This work indicates 
that \sell\ is a broader framework for handling such systems, since it can encode HyLL's rules naturally and directly. However, the simplicity of HyLL may be of interest  for specific purposes, such as building tools for 
diagnosis in biomedicine. 
And we can use the encoding of HyLL into LL in order to perform automatic proofs of properties of
  systems encoded in HyLL, for example. 
%
%

Formal proofs in HyLL were implemented in \cite{deMaria-Despeyroux-Felty:14-fmmb}, in the 
Coq  proof assistant.
It would be interesting to extend the implementations of HyLL given there to \sell.
Such an interactive proof environment
would enable both formal studies of encoded systems in \sell\ and formal meta-theoretical study of \sell\ itself.

We may pursue the goal of using HyLL/\sell\ for further applications. That might include neuroscience, a young and promising
science where many hypotheses are provided and need to be verified. 

 
\bibliographystyle{alpha}
\bibliography{references}


 \newpage

\appendix

\section{One Side Focused Proof System for Linear Logic} \label{app:ll} 

\begin{figure}[!h]
\noindent{\em Negative rules}
$$\frac{\Up\Psi\Delta L}{\Up\Psi\Delta{\bot,L}}\ [\bot]
  \qquad
  \frac{\Up{\Psi}{\Delta}{F,G,L}}{\Up{\Psi}{\Delta}{F\lpar G,L}}\ [\lpar]
  \qquad
  \frac{\Up{\Psi,F}{\Delta}{L}}{\Up{\Psi}{\Delta}{\quest F,L}}\ [\quest]
$$
$$\frac{}{\Up{\Psi}{\Delta}{\top,L}}\ [\top]
  \qquad
  \frac{\Up{\Psi}{\Delta}{F,L}\quad \Up{\Psi}{\Delta}{G,L}}
       {\Up{\Psi}{\Delta}{F\with G,L}}
     \ [\with]
  \qquad
  \frac{\Up{\Psi}{\Delta}{F[y/x],L}}{\Up{\Psi}{\Delta}{\forall x.F,L}}\ [\forall]
$$

\smallskip\noindent{\em Positive rules}
$$\frac{}{\Down{\Psi}{\cdot}{\one}}\ [\one]
  \qquad
  \frac{\Down{\Psi}{\Delta_1}{F}\quad \Down{\Psi}{\Delta_2}{G}}
       {\Down{\Psi}{\Delta_1,\Delta_2}{F\otimes G}} \ [\otimes]
  \qquad
  \frac{\Up{\Psi}{\cdot}{F}}{\Down{\Psi}{\cdot}{\bang F}}\ [\bang]
$$
$$\frac{\Down{\Psi}{\Delta}{F_1}}{\Down{\Psi}{\Delta}{F_1\oplus F_2}}
  \ [\oplus_l]
  \qquad
  \frac{\Down{\Psi}{\Delta}{F_2}}{\Down{\Psi}{\Delta}{F_1\oplus F_2}}
  \ [\oplus_r]
  \qquad
  \frac{\Down{\Psi}{\Delta}{F[t/x]}}{\Down{\Psi}{\Delta}{\exists x.F}}
              \ [\exists]
$$

\smallskip\noindent{\em Identity, Decide, and Reaction rules}
$$\frac{}{\Down{\Psi}{A}{\nng{A}}}\ [I_1]
  \qquad
  \frac{}{\Down{\Psi,A}{\cdot}{\nng{A}}}\ [I_2]
  \qquad
  \frac{\Down{\Psi}{\Delta}{F}}{\Up{\Psi}{\Delta,F}{\cdot}}  \ [D_1]
  \qquad
   \frac{\Down{\Psi,F}{\Delta}{F}}{\Up{\Psi,F}{\Delta}{\cdot}} \ [D_2]
$$
In $[I_1]$ and  $[I_2]$, $A$ is atomic; 
in $[D_1]$ and  $[D_2]$, $F$ is not an atom.
$$\begin{array}{c@{\quad}l}
{\displaystyle
\frac{\Up{\Psi}{\Delta,F}{L}}{\Up{\Psi}{\Delta}{F,L}}}\ [R\Uparrow]&
\hbox{\quad provided that $F$ is positive or an atom}\\
\noalign{\medskip}
{\displaystyle
\frac{\Up{\Psi}{\Delta}{F}}{\Down{\Psi}{\Delta}{F}}}\ [R\Downarrow]&
\hbox{\quad provided that $F$ is negative}
\end{array}
$$
\caption{Focused proof linear logic system LLF.}\label{system:LLF}
\end{figure} 
\newpage
\section{HyLL Sequent System} \label{appendix:hyll.seq}


%


\textit{Judgmental rules}

$\begin{array}{c}
  {\Gamma ; p(\vec{t}) ~@~ w \vdashseq p(\vec{t}) ~@~ w} ~[init]
  \qquad
  \dfrac{\Gamma, A ~@~ u; \Delta, A ~@~ u \vdashseq C ~@~ w}
        {\Gamma, A ~@~ u ; \Delta \vdashseq C ~@~ w} ~[copy]
\end{array}$

\textit{Multiplicative rules}
\\

\resizebox{.8\textwidth}{!}{
$\begin{array}{c}
     \dfrac{\Gamma ; \Delta \vdashseq A ~@~ w  \quad  \Gamma ; \Delta' \vdashseq B ~@~ w}
           {\Gamma ; \Delta, \Delta' \vdashseq A \otimes B ~@~ w} [\otimes R]
     \quad
     \dfrac{\Gamma ; \Delta, A ~@~ u, B ~@~ u \vdashseq C ~@~ w}
           {\Gamma ; \Delta, A \otimes B ~@~ u \vdashseq C ~@~ w} [\otimes L] 
     \\\\
     {\Gamma; . \vdashseq \mathbf{1} ~@~ w} ~[\mathbf{1} R]
     \quad 
     \dfrac{\Gamma ; \Delta \vdashseq C ~@~ w}
           {\Gamma ; \Delta, \mathop{1} @~ u \vdashseq C ~@~ w} [\mathop{1} L] 
     \\\\
     \dfrac{\Gamma ; \Delta,~ A ~@~ w \vdashseq B ~@~ w}
           {\Gamma ; \Delta \vdashseq A \limp B ~@~ w} [\limp R]
     \qquad
     \dfrac{\Gamma ; \Delta \vdashseq A ~@~ u \quad \Gamma ; \Delta', B ~@~ u \vdashseq C ~@~ w}
           {\Gamma ; \Delta, \Delta', A \limp B ~@~ u \vdashseq C ~@~ w} [\limp L] 
\end{array}$
}

\textit{Additive rules}

$\begin{array}{c}
     {\Gamma ; \Delta \vdashseq T ~@~ w} ~[T~ R]
     \qquad
     {\Gamma ; \Delta, \mathbf{0} ~@~ u \vdashseq C ~@~ w} ~[\mathbf{0} L]
     \\\\
     \dfrac{\Gamma ; \Delta \vdashseq A ~@~ w \qquad \Gamma ; \Delta \vdashseq B ~@~ w}
           {\Gamma ; \Delta \vdashseq A \mathbin{\&} B ~@~ w}  ~[\mathbin{\&} R]
     \qquad
     \dfrac{\Gamma ; \Delta, A_i ~@~ u \vdashseq C ~@~ w}
           {\Gamma ; \Delta, A_1 \mathbin{\&} A_2 ~@~ u \vdashseq C ~@~ w} 
             [\mathbin{\&} L_i]
     \\\\
     \dfrac{\Gamma ; \Delta \vdashseq A_i ~@~ w}
           {\Gamma ; \Delta \vdashseq A_1 \oplus A_2 ~@~ w} [\oplus R_i]
     \quad
     \dfrac{\Gamma ; \Delta, A ~@~ u \vdashseq C ~@~ w  \quad 
            \Gamma ; \Delta, B ~@~ u \vdashseq C ~@~ w}
           {\Gamma ; \Delta, A \oplus B ~@~ u \vdashseq C ~@~ w} ~[\oplus L]
\end{array}$

\textit{Quantifier rules} 
$$\begin{array}{c}
  \dfrac{\Gamma ; \Delta \vdashseq A ~@~ w}
        {\Gamma ; \Delta \vdashseq \forall \alpha.~ A ~@~ w}  ~[\forall R^\alpha] 
  \qquad
  \dfrac{\Gamma ; \Delta, A [\tau / \alpha] ~@~ u \vdashseq C ~@~ w} 
        {\Gamma ; \Delta, \forall \alpha.~ A ~@~ u \vdashseq C ~@~ w} ~[\forall L] 
  \\\\
  \dfrac{\Gamma ; \Delta \vdashseq A [\tau / \alpha] ~@~ w}
        {\Gamma ; \Delta \vdashseq \exists \alpha.~ A ~@~ w} ~[\exists R]
  \qquad
  \dfrac{\Gamma ; \Delta, A ~@~ u \vdashseq C ~@~ w}
        {\Gamma ; \Delta,\exists \alpha.~ A ~@~ u \vdashseq C ~@~ w} ~ [\exists L^\alpha]
  \\\\ 
  \begin{minipage}{0.9\linewidth}  
     In $\forall R^\alpha$ and $\exists L^\alpha$, $\alpha$ is assumed to be fresh
     with respect to $\Gamma$, $\Delta$, and $C$.\\
     In $\exists R$ and $\forall L$, $\tau$ stands for a term or world, as
     appropriate.
  \end{minipage}
\end{array}$$
\textit{Exponential rules}
$$\begin{array}{c}
  \dfrac{\Gamma ; . \vdashseq A ~@~ w}{\Gamma ; . \vdashseq {! A} ~@~ w}  ~[! R]
  \qquad
  \dfrac{\Gamma, A ~@~ u ; \Delta \vdashseq C ~@~ w}
        {\Gamma ; \Delta, {! A} ~@~ u \vdashseq C ~@~ w} ~[! L]
\end{array}$$
\textit{Hybrid connectives}
$$\begin{array}{c}
  \dfrac{\Gamma ; \Delta \vdashseq A ~@~ u} 
        {\Gamma ; \Delta \vdashseq (A ~\at~ u) ~@~ w} ~[\at~ R]
 \qquad
  \dfrac{\Gamma ; \Delta, A ~@~ u \vdashseq C ~@~ w} 
        {\Gamma ; \Delta, (A ~\at~ u) ~@~ v \vdashseq C ~@~ w} ~[\at~ L]
 \\\\
  \dfrac{\Gamma ; \Delta \vdashseq A [w / u] ~@~ w} 
        {\Gamma ; \Delta \vdashseq \downarrow u. A ~@~ w} ~[\downarrow R]
  \qquad
  \dfrac{\Gamma ; \Delta, A [v / u] ~@~ v \vdashseq C ~@~ w}
        {\Gamma ; \Delta, \downarrow u. A ~@~ v \vdashseq C ~@~ w} ~[\downarrow L]
\end{array}$$
%
%
 %

\newpage
\section {\sellU Sequent System} \label{appendix:sell}
\begin{figure}[h]
$$
\infer[\text{$\with$}]{\vdash \mathcal{K} : \Gamma \Uparrow L, A
\with B}{\vdash \mathcal{K} : \Gamma \Uparrow L, A \;\;\; \vdash
\mathcal{K} : \Gamma \Uparrow L, B}
\quad
\infer[\text{$\lpar$}]{\vdash \mathcal{K} : \Gamma \Uparrow L, A
\lpar B}{\vdash \mathcal{K} : \Gamma \Uparrow L, A, B}
$$
$$
\infer[\text{$\bot$}]{\vdash \mathcal{K} : \Gamma \Uparrow L,
\bot}{\vdash \mathcal{K} : \Gamma \Uparrow L}
\qquad
\infer[\text{$\top$}]{\vdash \mathcal{K} : \Gamma \Uparrow L, \top}{}
\qquad
\infer[\text{$?^l$}]{\vdash \mathcal{K} : \Gamma \Uparrow L, ?^l
A}{\vdash \mathcal{K} +_l A : \Gamma \Uparrow L}
$$
$$
\infer[\text{$\forall$}]{\vdash \mathcal{K} : \Gamma \Uparrow L,
\forall
x.A}{\vdash \mathcal{K} : \Gamma \Uparrow L, A\{c/x\}}\qquad
\infer[\forallLoc_R]{\vdash \mathcal{K} :  \Gamma \Uparrow\forallLoc l_x:a. G}
{
\vdash \mathcal{K} :  \Gamma \Uparrow G[l_e/l_x]
} 
$$
$$
\infer[\text{$\oplus_i$}]{\vdash \mathcal{K} : \Gamma \Downarrow A_1
\oplus A_2}{\vdash \mathcal{K} : \Gamma \Downarrow A_i}
\qquad
\infer[\text{$\otimes$, given $(\mathcal{K}_1 =
\mathcal{K}_2)|_{\Uscr}$}]{\vdash \mathcal{K}_1
\otimes \mathcal{K}_2 : \Gamma, \Delta \Downarrow A \otimes B}{\vdash
\mathcal{K}_1 : \Gamma \Downarrow A \;\;\; \vdash \mathcal{K}_2 : \Delta
\Downarrow B}
$$
$$
\infer[\text{1, given $\mathcal{K}[\mathcal{I} \setminus \Uscr] =
\emptyset$}]{\vdash \mathcal{K} : \cdot \Downarrow 1}{}
\quad
\infer[\text{$\exists$}]{\vdash \mathcal{K} : \Gamma \Downarrow \exists
x.A}{\vdash \mathcal{K} : \Gamma \Downarrow A\{t/x\}}
\quad 
\infer[\existsLoc_L]{\vdash \mathcal{K} :  \Gamma\Downarrow  \existsLoc l_x:a.G }
{\vdash \mathcal{K} :  \Gamma
\Downarrow   G[l/l_x] 
}
$$
$$
\infer[\text{$!^l$, given $\mathcal{K}[\{x\ |\ l \npreceq x \wedge x
\notin \Uscr\}] = \emptyset$}]{\vdash \mathcal{K} : \cdot \Downarrow
!^l A}{\vdash \mathcal{K} \leq_l : \cdot \Uparrow A}
$$
$$
\infer[\text{I, given $A_t \in (\Gamma \cup
\mathcal{K}[\mathcal{I})$ and $(\Gamma \cup \mathcal{K}[\mathcal{I}
\setminus \Uscr]) \subseteq \{A_t\} $}]{\vdash \mathcal{K}
: \Gamma \Downarrow A_t^\bot}{}
$$
$$
\infer[\text{$D_l$, given $l \in \Uscr$}]{\vdash
\mathcal{K} +_l P : \Gamma \Uparrow \cdot}{\vdash \mathcal{K} +_l P :
\Gamma \Downarrow P}
\qquad
\infer[\text{$D_l$, given $l \notin \Uscr$}]{\vdash \mathcal{K} +_l P :
\Gamma \Uparrow \cdot}{\vdash
\mathcal{K} : \Gamma \Downarrow P}
$$
$$
\infer[\text{$D_1$}]{\vdash \mathcal{K} : \Gamma, P \Uparrow
\cdot}{\vdash \mathcal{K} : \Gamma \Downarrow P}
\qquad
\infer[\text{$R\Downarrow$}]{\vdash \mathcal{K} : \Gamma \Downarrow
N}{\vdash \mathcal{K} : \Gamma \Uparrow N}
\qquad
\infer[\text{$R\Uparrow$}]{\vdash \mathcal{K} : \Gamma \Uparrow L,
S}{\vdash \mathcal{K} : \Gamma, S \Uparrow L}
$$
\vspace{-2mm}
\caption{Focused linear logic system with (quantified) subexponentials. 
Here, $L$ is a list of formulas,
$\Gamma$ is a multi-set of formulas and positive literals, $A_t$ is an
atomic formula, $P$ is a non-negative literal,
$S$ is a positive
literal or formula and $N$ is a negative formula.}
\label{figure:sellf}
\vspace{-3mm}
\[
\begin{array}{l@{\qquad}l}
\bullet~(\mathcal{K}_1 \tensor \mathcal{K}_2) [\tsl{i}] = \left\{
\begin{array}{ll}
 \mathcal{K}_1[\tsl{i}] \cup \mathcal{K}_2[\tsl{i}] & \hbox{ if }
\tsl{i} \notin \mathcal{U}\\
 \mathcal{K}_1[\tsl{i}]  & \hbox{ if } \tsl{i} \in \mathcal{U}
\end{array}
\right.
& 
\bullet~\mathcal{K}[\mathcal{S}] =
\bigcup\{\mathcal{K}[\tsl{i}]\;|\;\tsl{i}\in \mathcal{S}\}\\[15pt]
\bullet~(\mathcal{K} +_l A) [\tsl{i}] = \left\{
\begin{array}{ll}
 \mathcal{K}[\tsl{i}] \cup \{A\} & \hbox{ if } \tsl{i} = l\\
 \mathcal{K}[\tsl{i}]  & \hbox{ otherwise }
\end{array}
\right.
&
\bullet~ \mathcal{K} \leq_i[\tsl{l}] = \left\{
\begin{array}{ll}
 \mathcal{K}[\tsl{l}] & \hbox{ if } i \preceq \tsl{l}\\
 \emptyset & \hbox{ if } i \npreceq \tsl{l} 
\end{array}
\right.
\end{array}
\]
\[
\bullet~ (\mathcal{K}_1 \star \mathcal{K}_2)\mid_\mathcal{S}
\textrm{ is true if and only if }(\mathcal{K}_1[\tsl{j}]
\star \mathcal{K}_2[\tsl{j}])
\]
\vspace{-3mm}
\caption{Specification of operations on contexts. Here, 
$\tsl{i} \in I$, $\tsl{j} \in \Sscr$, $\Sscr \subseteq I$, and the 
binary connective $\star \in \{=, \subset, \subseteq\}$.}
\label{Fig:Contexts}
\vspace{-4mm}
\end{figure}
 
\newpage
\section{Adequacy Proofs} \label{appendix:proofs}

\subsection{CTL in HyLL}
\label{app:ctlhyll}
\begin{proposition}\label{prop:state-trans}
Let $\texttt{V}$ be a set of variables and  $\Rscr=\{r_1,...,r_m\}$ be a set of transition rules on $\texttt{V}$. Then,  $\stateS \rede{r_i} \stateS'$ iff the 
sequent 
\[ \os r_1\cs @0, \cdots , \os r_n\cs @0; 
\os \stateS \cs @w \vdash \delay{1}\os \stateS' \cs @w
\] is provable in HyLL. 
\end{proposition}
\begin{proof}  We will use the focused version of the HyLL system \cite{ChaudhuriDespeyroux:14}  and assume that atoms have positive bias. Observe that $\os \stateS \cs  @w$ and $ \delay{1}\os \stateS' \cs @w$ will be decomposed in the negative phase,  giving rise to a sequent of the shape
$$ \os r_1\cs @0, \cdots ,\os r_n\cs @0;
 \mathtt{s}_1@w,\ldots,  \mathtt{s}_n@w\vdash \mathtt{s}'_1@w.1\otimes\ldots  \otimes\mathtt{s}'_n@w.1$$ where $\mathtt{s}_i,
 \mathtt{s}_i' \in\{\presS{a_i}, 
 \absS{a_i}\}$. Since each $\mathtt{s}_i$ is a literal in the world $w$ and the literals in $\mathtt{s}'_i$ belong to the world $w.1$, it is not possible to focus on $\mathtt{s}'_1@0\otimes\ldots  \otimes\mathtt{s}'_n@0$. Hence,  the only possibility of proceeding with the proof is by focusing on one of the formulas $\os r_i\cs @0$. Due to the shape of $\os r_i\cs @ 0$, in one focused step, we consume {\em all} atoms from $\os \mathtt{s}\cs @ w$  and add to the context the formula  $\os \mathtt{s'}\cs @ w.1$. This  mimics exactly the 
 transition $\stateS \rede{r_i} \stateS'$ . Now we can focus on the right of the sequent and the proof finishes. Note that focusing again in a formula $\os r_i\cs @ 0$ is pointless since the atoms will be produced 
 at the world $w+n$ ($n>1$) and hence, never matching the formulas on the right. 
\end{proof}

\paragraph{CTL Formulas in HyLL}
The next proposition considers only the encoding of the fragment of CTL into HyLL
presented in Section \ref{sec:tsHyLL}. 
\begin{proposition}
The CTL formula $F$ holds at state $\stateS$ iff 
$\os \Rscr\cs@0; \os\stateS\cs @ w\vdash \encCTL{ F} @ w$ is provable in HyLL. 
\end{proposition}
\begin{proof}
($\Rightarrow$) By induction on the structure of $F$.
For the base case, if the state $\stateS$ satisfies the state formula $p$, it is easy to show that  
the sequent $\os \stateS \cs @ u \vdash \os P \cs @ u$ is provable in HyLL. 
If $\ctlS{\stateS}{\Rscr} \tE  \tF ~ F$, then there is a path $\stateS = \stateS_1 \rede{r-{i_1}} \stateS_2 \cdots \stateS_{m-1} \rede{r-{i_{m-1}}} \stateS_m$ 
where $F$ holds at $\stateS_m$. 
Hence,  at each step, we choose  a rule from $\Rscr$ until a $F$-state is reached. 
In HyLL, looking a derivation bottom-up, we focus on one of the rules and transform the 
state $\os  \stateS_i\cs @w$ into the state $\os  \stateS_{i+1}\cs @{w.1}$.
The result follows by repeated applications of Proposition \ref{prop:state-trans}. The cases for $\tE[F \tU G]$ and $\tE\tX F$ follow similarly. Finally, the cases for $\wedge$ and $\vee$ follow immediately by inductive hypothesis. 

($\Leftarrow$) Using the focusing discipline, we shall show that each focused step corresponds exactly to a ``step'' in the deduction of $\ctlS{\stateS}{\Rscr}  F$. In order to ease the proof, we shall consider an encoding slightly different. Let $\pdelay{F}=F\otimes \one$ and $\ndelay{F}=\one\limp F$ respectively. 
Observe that
$\pdelay{F}\equiv\ndelay{F}\equiv F$. Hence, we only introduce positive or negative delays that allow us to focus / disallow focusing in a derivation. The proposed encoding is:
\[
\begin{array}{lll}
\encCTL{\stateS} & = & \os \stateS\cs
\\
\encCTL{ F \wedge G} & = & \pdelay{\encCTL{F} \with \encCTL{G}}
\\
\encCTL{ F \vee G} & = & \ndelay{\encCTL{F}} \oplus \ndelay{\encCTL{G}}
\\
\encCTL{\tE[F \tU G]} & =& \encCTL{F} \tU~ \encCTL{G}
\\
\encCTL{\tE\tX F} & =& \pdelay{\delay{1}\encCTL{F}}
\\
\encCTL{\tE\tF F}& =& \Diamond\encCTL{F} \\
\end{array}
\]


 Consider the sequent $\os \Rscr\cs@0; \os\stateS\cs @ w\vdash \encCTL{ F} @ w$.
 We have two choices: 1) we focus on one of the rules in $\Rscr$ and we transform, in one focused step, the state $\texttt{s}$ into the state $\stateS'$; or 2), we focus on the formula on the right. In the first case, we already showed that this action mimics exactly the transition $\stateS \rede{r_i} \stateS'$.  In the second case, we note that the formula on the right must have the following shape:
\[
\begin{array}{lll}
F &::=& S \mid \one\otimes (F \with F) \mid F \oplus F \mid \downarrow u \ (F ~\at~ u.1) \mid \downarrow u\  (\exists w.F ~\at~u.w) \mid\\
& &
\downarrow u\  \exists v. (F ~\at~u.v \with \forall w\prec v. F ~\at~u.w)
\end{array}
\]
where $S$ is of the shape $\mathtt{s}_1@w\otimes\ldots  \otimes\mathtt{s}_n@w$ (the encoding of a state formula). The other  cases represent, respectively, the encoding of CTL formulas of the shape $F\wedge F$, $F \vee F$, $\tE\tX F$, $\tE\tF F$ and $\tE[F \tU F]$. 
In a negative phase, the only connectives we can introduce, if any,  are the hybrid ones ($\downarrow$ and $\at$). This is a bureaucratic step  allowing us to fix the formulas at the ``current'' world as in 
\[
\infer=[\at_R, \downarrow_R]{\Gamma;  \Delta \vdash \downarrow x (F~\at~y) @ w}{\Gamma;  \Delta \vdash F[x/w] @ y}
\]
Hence, when focusing on the right, we fall in the following cases. If we focus on:
\begin{itemize}
 \item $S$, the context must already have the atoms, at the right world, to prove the conjunction of atoms in $S$. This corresponds to proving that the state $\texttt{s}$ satisfies the state property $S$. 
 \item $\one \otimes (F \with G)$, we prove $\one$ and we lose focusing in $F \with G$. Hence, after a negative phase, we have a derivation proving $F$ and another proving $G$. This corresponds exactly to the step of proving $F$  and $G$ in CTL. 
 
 \item $F \oplus G$, we chose one of the branches and then, we lose focusing again (due to the negative delay in the encoding). This corresponds to proving either $F$ or $G$ in CTL. 
\item  $\pdelay{\delay{1}F} $, we lose focusing again and we obtain, on the right, $F$  fixed at the world $w+1$. This mimics the step of proving $F$ in the next time-unit to show that $\tE\tX F$ holds in CTL. 

\item $\exists w.F ~\at~u.w$, we choose a world $w$ and we lose focusing (due to $\at$). This corresponds in CTL to proving $\tE \tF F$
by showing that  there exists   a future world ($u+w$) where $F$ holds. 

\item The case of $\tE[F \tU G]$ is similar to the previous one. 
\end{itemize}

\end{proof}

\subsection{CTL in $\mu$MALL}

{\bf Theorem \ref{th:ctl_mumall}}. 
Let $\texttt{V}=\{a_1,...,a_n\}$ be a set of propositional variables, $\Rscr$ be a set of transition rules on $\texttt{V}$ and $F$ be a CTL formula. Then, $\ctlS{\stateS}{\Rscr}  F$  iff
the sequent 
$\vdash\os \stateS\cs, \encCTLR{F} $ is provable in $\mu$MALL. 
\begin{proof}


 ($\Rightarrow$) We proceed by induction on the structure of the formula. 
 The base case of a state formula $P$ is immediate (from the encoding $\os \stateS\cs$). 
 The cases for $\wedge$ and $\vee$ are  easy consequences from the inductive hypothesis. 
 
\noindent {\bf Cases $\tA \tX$ and $\tE\tX$. }
Let us note that,  given two different states $\stateS$ and $\stateS'$: 
\begin{itemize}
 \item the sequents $\vdash \os \stateS\cs , \positiveS(\stateS)$ and  $\vdash \os \stateS\cs, \negativeS(\stateS')$ are both provable. 
 \item the sequents $\vdash \os \stateS\cs , \negativeS(\stateS)$ and  $\vdash \os \stateS\cs, \positiveS(\stateS')$ are not provable. 
\end{itemize}

This means that, in a  context containing a formula $\os \stateS\cs$, we can always prove that a given rule $r_i \in \Rscr$ is firable or not.

 Consider the case $\tA \tX F$. In a negative phase, we obtain the following derivation:\\
 
 \resizebox{\textwidth}{!}{
$ 
\infer[\with \qquad]{\vdash \os \stateS\cs , \bigwith\limits_{\stateS\to \stateS'\in \Rscr } \left(  {\negativeS(\stateS)} \oplus ( {\positiveS(\stateS) } \otimes \left( { \os \stateS' \cs } \invamp~ \phi\right)  \right) }{
 \vdash  \os \stateS\cs , \negativeS(\stateS_1) \oplus (\positiveS(\stateS_1) \otimes ( \os \stateS_1'\cs  \invamp~ \phi)
  \quad ...\quad
  \vdash \os \stateS\cs , \negativeS(\stateS_m) \oplus (\positiveS(\stateS_m) \otimes ( \os \stateS_m'\cs  \invamp~ \phi)
}
$
}

\noindent
where $\phi = \encCTLR{F}$. In each case, for every premise we have to start a positive phase and we have to choose between $\negativeS(\stateS_i)$ and $\positiveS(\stateS_i)$. In the first case, if the rule is not fireable, the proof ends. In the second case, we obtain a derivation of the shape:
\[
 \infer=[\otimes,\invamp]{\vdash \os \stateS \cs, \positiveS(\stateS_i)   \otimes (\os \stateS_i'\cs\invamp~ \phi) }{
  \deduce{\vdash \os \stateS_i'\cs, \phi}{}
 }
\]
 and the positive phase ends. By inductive hypothesis, the sequent $\vdash \os \stateS_i'\cs, \phi$ is provable. 
  The case $\tE \tX F$ is similar. 
 
\noindent{\bf Cases for the least fixed point operators.} 
If 
$\tA\tF F$ holds in CTL at state $\stateS$, then, in all paths  starting at $\stateS$, there is a reachable state $\stateS'$ such that $F$ holds in that state. Let $\stateS = \stateS_1 \rede{}   \cdots \rede{} \stateS_n=\stateS'$ be one of such paths and consider the following derivation: \\

\resizebox{\textwidth}{!}{
$
\infer=[\mu,\oplus,\with \qquad]{\vdash \os \stateS\cs , \mu B }{
 \vdash  \os \stateS\cs , \negativeS(\stateS_1) \oplus (\positiveS(\stateS_1) \otimes ( \os \stateS_1'\cs  \invamp~ \mu B)
  \quad ...\quad
  \vdash \os \stateS\cs , \negativeS(\stateS_m) \oplus (\positiveS(\stateS_m) \otimes ( \os \stateS_m'\cs  \invamp~ \mu B)
}
$
}
\ \\ The  premises correspond to proving whether the rule $r_i$ is fireable or not. If $r_i:\stateS_i \to\stateS_i' $ is fireable, we observe a derivation of the shape:
\[
\infer[\oplus]{ \vdash \os \stateS\cs, \negativeS(\stateS_i)  \oplus (\positiveS(\stateS_i)   \otimes (\os \stateS_i'\cs\invamp~ \mu B) )}{
 \infer[\otimes,\invamp]{\vdash \os \stateS \cs, \positiveS(\stateS_i)   \otimes (\os \stateS_i'\cs\invamp~ \mu B) }{
  \deduce{\vdash \os \stateS_i'\cs, \mu B }{}
 }
 }
\]
where $\stateS$ becomes $\stateS_i'$ and, from that state, $\mu B$ must be proved. 
Hence,  we can show that $\os \stateS_n\cs$ will be eventually added to the context. By inductive hypothesis, 
the sequent $\vdash \encCTLR{F}, \os \stateS_n\cs$ is provable and then, 
$\vdash \os \stateS_n\cs, \mu B$ is provable (by choosing $\encCTLR{F}$ in the disjunction  $\encCTLR{\tA \tF F} =  \mu Y.  \encCTLR{F}~   \oplus     \Phi$). 

The other cases for least fixed point operators   follow similarly. 

\noindent{\bf Cases for the greatest fixed point operators.}
Consider now the formula $\tA \tG F$. If this formula holds at $\stateS$, 
then  $\stateS$ must satisfy $F$ and all path starting from $\stateS$
must also satisfy $\tA \tG F$. Let 
\[
S = \{ s \mid \ctlS{s}{\Rscr}  F \mbox{ and, for all $s'$, if } s \rede{} s' \mbox{, then } s' \in S\} 
\]
be the greatest set of states containing  $\stateS$. Note that the greatest fixed point in the (CTL) definition of $\tA\tG$ computes exactly that set. 

Let $S$ above be the set $\{s_1,...,s_n\}$ and $I= \os s_1\cs^\perp \oplus \cdots \oplus \os s_n\cs^\perp$ \footnote{We note that from a finite set of propositional variables,  the set of states (containing, $\absS{x}$ or $\presS{x}$)  is finite. Hence,   any infinite path in such LTS must have a loop. }. 
We shall show that, for any $s\in S$,  the sequent $\vdash \os s \cs, \encCTLR{\tA \tG F}$ is provable using $I$ as inductive invariant. 

Once the rule $\nu$ is applied, we have to prove two premises:
\begin{enumerate}
 \item {\bf Premise } $\mathbf{\vdash \os s\cs, I}$. This sequent is easy by  choosing  $\os s\cs^\perp$ from $I$. 
 \item {\bf Premise } $\mathbf{\vdash B\  I, I^\perp}$. The $\bigwith\limits_{s\in S} \os s \cs$ formula in $I^\perp$ forces us to prove several cases. More precisely, for each $s\in S$, we have to prove $\vdash B I, \os s\cs$. Consider the following derivation:
 \[
 \infer[\with]{\vdash  \phi \with R_1 \with \cdots \with R_n, \os s \cs}{
  \deduce{\vdash \phi, \os s\cs}{}
  &
  \deduce{\vdash R_1, \os s\cs}{}
  & \cdots
  & 
  \deduce{\vdash R_n, \os s\cs}{}
 }
 \]
 where $\phi = \encCTLR{F}$ and $R_i =   \negativeS( \stateS_i ) \oplus ( {\positiveS( \stateS_i)} \otimes \left( { \os \stateS_i' \cs } \invamp~ I \right)$. 
 Again we have several cases to prove. 
 
 The first sequent $\vdash \phi, \os s\cs$ follows from inductive hypothesis. 
 
 If the rule $r_i$ is not fireable at state $s$, the sequent $\vdash \os s\cs, R_i$ is provable (by choosing $\negativeS(\stateS_i)$). 
 
 If $r_i$ is  fireable at state $s$, we then have, in a focused step, the following derivation: 
 \[
  \infer=[\oplus,\otimes,\with]{\vdash R_i, \os s\cs}{
   \deduce{\vdash \os s'\cs, I}{}
  }
 \]
 Since $S$ is closed under $\rede{}$, it must be the case that $s' \in S$ and then, the sequent $\vdash \os s'\cs, I$ is provable. 
\end{enumerate}

The case $\tE\tG$ is similar. 

\noindent($\Leftarrow$) Due to focusing, we can show that the derivations  in the $\Rightarrow$ part are the only way to proceed during a proof in $\mu$MALL. Hence, we match exactly a ``step'' in the deduction of $\ctlS{\stateS}{\Rscr}  F$. 

The only 
interesting case is the one of the greatest fixed point operator. Consider the CTL formula $\tA\tG F$  and assume that we have a proof of the sequent $\vdash \os \stateS \cs , \nu B $
with invariant $I_x$. This means that we have a proof of the sequent $\vdash \os \stateS \cs, I_x$. Moreover, due to the shape of $B$, we must also have a  proof of $\vdash \os \stateS' \cs, I_x$ for any reachable state $\stateS'$. Then, we can show that there is a proof of $\vdash I_x, \bigwith\limits_{s\in S} \os s\cs$ where $\stateS \in S$ and all reachable state $\stateS'$ (from $\stateS$), is also in $S$. Let $I$ be the invariant in the proof of the $\Rightarrow $ part. Note that $I^\perp = \bigwith\limits_{s\in S} \os s\cs$ and then, we have a proof of $\vdash I_x, I^\perp$ (i.e., $\vdash I \limp I_x$). This shows that $I$ is greater than  $I_x$ and then, we also have a proof of 
$\vdash \os \stateS \cs , \nu B $ using $I$. The result follows from a  derivation similar to the one used in the proof of the $\Rightarrow$ part. 
\end{proof}


\end{document}